\documentclass[letterpaper]{article}
\usepackage[margin=1in]{geometry}

\usepackage{ifpdf}
\ifpdf
    \usepackage[pdftex]{graphicx}
    \usepackage[update]{epstopdf}
\else
	\usepackage{graphicx}
\fi

\usepackage{booktabs}
\usepackage{amsmath,amsthm,amsfonts}%,dsfont,bm}
\usepackage{physics}
\usepackage{textgreek}
\usepackage{enumitem}
\usepackage{multirow}
\usepackage[ruled,vlined]{algorithm2e}
\usepackage{xcolor}
\usepackage{soul}
\usepackage{tikz}
\usetikzlibrary{positioning}

\definecolor{Mycolor}{rgb}{.2,.2,.8}

\newtheorem{definition}{Definition}

\newtheorem{theorem}{Theorem}

\newtheorem{corollary}{Corollary}
\newtheorem{lemma}[theorem]{Lemma}

\newcommand{\Cc}{\mathcal{C}}

\newcommand{\Hc}{\mathcal{H}}

\newcommand{\Nc}{\mathcal{N}}
\newcommand{\Pc}{\mathcal{P}}

\newcommand{\Tc}{\mathcal{T}}

\newcommand{\Wc}{\mathcal{W}}
\newcommand{\Xc}{\mathcal{X}}
\newcommand{\Yc}{\mathcal{Y}}
\newcommand{\Zc}{\mathcal{Z}}

\newcommand{\Nb}{\mathbb{N}}
\newcommand{\Eb}{\mathbb{E}}
\newcommand{\Rb}{\mathbb{R}}

\newcommand{\argmax}{\arg\max}
\newcommand{\argmin}{\arg\min}

\newcommand{\gen}{\mathrm{gen}}

\begin{document}
\title{Stochastic Chaining and Strengthened Information-Theoretic Generalization Bounds}

\author{Ruida Zhou \and Chao Tian \and Tie Liu}
\date{}
\maketitle

\begin{abstract}
We propose a new approach to apply the chaining technique in conjunction with information-theoretic measures to bound the generalization error of machine learning algorithms. Different from the deterministic chaining approach based on hierarchical partitions of a metric space,  previously proposed by Asadi et al., we propose a stochastic chaining approach, which replaces the hierarchical partitions with an abstracted Markovian model borrowed from successive refinement source coding. This approach has three benefits over deterministic chaining: 1) the metric space is not necessarily bounded, 2) facilitation of subsequent analysis to yield more explicit bound, and 3) further opportunity to optimize the bound by removing the geometric rigidity of the partitions. The proposed approach includes the traditional chaining as a special case, and can therefore also utilize any deterministic chaining construction. We illustrate these benefits using the problem of estimating Gaussian mean and that of  phase retrieval. For the former, we derive a bound that provides an order-wise improvement over previous results, and for the latter we provide a stochastic chain that allows optimization over the chaining parameter. 
\end{abstract}

\section{Introduction}

\subsection{Motivation}

Bounding the generalization error of machine learning algorithms has been studied extensively in the literature. Classical results in this area, such as the VC dimension and the Rademacher complexity, focused mostly on capturing the structural constraints of the hypothesis class. Despite the considerable insights offered by such bounds, they do not explain well the performance of powerful learning algorithms such as deep neural networks \cite{zhang2016understanding}. There has been significant recent interest in finding information-theoretic bounds to capture more data-dependent and distribution-dependent structures which may lead to strengthened bounds.

Asadi et al. \cite{asadi2018chaining} introduced the chaining technique, which has traditionally been used in bounding random process, into the derivation of information-theoretic generalization bounds. The technique resolves the issue that certain unbounded mutual information quantity leads to a vacuous bound, and may also yield a tighter bound in general. The main idea behind the result in \cite{asadi2018chaining} can be summarized as follows. The generalization error can be viewed as a random process $\{X_t\}_{t\in\mathcal{T}}$ indexed by the hypothesis parameters. If $(\mathcal{T},d)$ is a bounded metric space under the metric $d$, then $\mathcal{T}$ can be divided into finer and finer partitions, with each coarse partition embedded into the next layer finer partition, and the partition cells having decreasing radius. The generalization error can then be represented by a sum of chained quantities, each relating to two adjacent partition layers. Since the partitions are becoming finer and finer, each of these decomposed quantities can be bounded more effectively, eventually resulting in an overall tighter bound. This approach is referred to as {\em chaining mutual information}.

Despite the success of the chaining mutual information approach, we observe several difficulties in applying the chaining technique in this manner, which motivated the current work:
\begin{itemize}
\item \textbf{Restriction on the metric space to be bounded:} This chaining approach assumes a bounded metric space $(\mathcal{T},d)$. However, even in some of the simplest settings, the parameter space may not be bounded (or impractical to assume the bound on $(\mathcal{T},d)$ is known). 
\item \textbf{Difficulty in computation: } Using these deterministic and hierarchical partitions, the information measures involved in the bounds can be difficult to compute or bound analytically.  
\item \textbf{Restrictions in the partitions:} The hierarchical partitions place certain unnecessary geometric constraints on the covering radius sequence of the required partitions, which can impact the bound. 
\end{itemize}

\begin{figure*}[tb]
\centering
\includegraphics[width=0.85\textwidth]{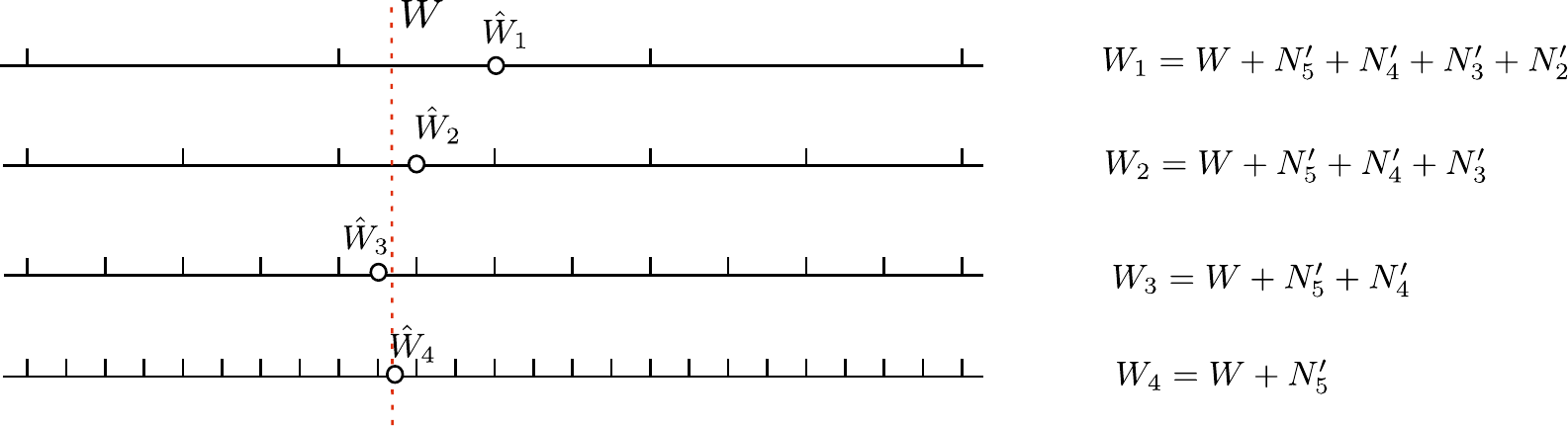}
\caption{Multilevel quantization of a random value $W$ using quantizers of different stepsize and the corresponding information-theoretic successive refinement source coding model. \label{fig:SR} }
\end{figure*}

To make these difficulties more concrete, consider the following two simple examples. 
\begin{itemize}

\item \textbf{Example-1}: The training samples are drawn $i.i.d.$ following a normal distribution with an unknown mean $\mu$, and the algorithm wishes to estimate this mean. Here the parameter space is $\mathcal{T}=\mathbb{R}$, which is unbounded under any meaningful metric, and particularly so for the natural Euclidean distance. Moreover, since the induced measure on $\mathcal{T}$ will not be uniform, computing the series sum of mutual information is rather difficult if not impossible. 
\item \textbf{Example-2}:  Let $Z:=(G_1,G_2)\sim \mathcal{N}(0,I_2)$ be standard normal vectors in $\mathbb{R}^2$. The learner needs to identify the phase of the vector through certain means, and the learned result is modeled as the true phase with certain additive noise. Here $\mathcal{T}$ is the bounded interval of the angle $[0,2\pi)$. A natural sequence of partitions is to reduce the stepsize by an integer factor $\gamma$. However, this would preclude any non-integer $\gamma$ values, which potentially makes the bound looser. 
\end{itemize}

\subsection{Main contribution: stochastic chaining}

The sequence of refining partitions of the metric space associated with the chaining technique is reminiscent of multilevel quantization in data compression. For example, a scalar source $W$ distributed on the real line can be quantized with stepsize of $2^{-k}$ for the $k$-th level quantization, resulting in its quantized representation $\hat{W}_k$. As the index $k$ increases, the stepsize reduces and the accuracy of the quantization improves; see the left side of  Fig. \ref{fig:SR} for an illustration.

The information-theoretic model for multilevel quantization is usually referred to as successive refinement source coding \cite{rimoldi1994successive,equitz1991successive}. Particularly useful to us is a stochastic abstraction in this framework. For example, assume there are a total of $K$-levels, then one possible  stochastic representation of the reconstruction $\hat{W}_k$ is $W_k$ that is written as 
\begin{align}
W_k=\alpha_k\left(W+\sum_{i=k+1}^{K+1} N'_i\right),
\end{align}
where $N'_i$'s are mutually independent random noises, also independent of $W$, and $\alpha_k$'s are certain fixed scalar coefficients; see the right side of  Fig. \ref{fig:SR}. It is seen that the relation among $W$ and $\{W_k\}_{k=1}^K$ is captured by the joint probability distribution among them, and we can measure the ``distance'' between $W$ and $W_k$ using $\Eb d(W,W_k)$, in contrast to the conventional chaining approach which uses the covering radius. 

The main idea of this work is that these abstracted stochastic versions of $\{W_k\}_{i=1}^K$ can be used to replace the partition-based quantized versions in bounding the generalization error. This new approach helps to resolve the difficulties mentioned above: firstly the restriction for the metric space to be bounded is naturally removed, and secondly, it helps to simplify the computation, and lastly, the abstract model can remove the geometric constraints in designing the hierarchical partitions in some cases. 

The proposed stochastic chaining approach essentially allows more flexible constructions of the chains than the more traditional deterministic chaining. One can attempt to further optimize the construction of stochastic chains based on the existing knowledge regarding the underlying metric space and the corresponding probability distribution for the given problem setting. On the other hand, when such knowledge is not available, we can safely fall back to the default construction of the original deterministic chaining partitions, which is essentially a special case of the stochastic chaining. 

We obtain two generalization bounds using stochastic chaining instead of the deterministic chaining in \cite{ asadi2018chaining}, built on the mutual information bound given in \cite{xu2017information} and the individual sample mutual inforamtion bound given in \cite{bu2020tightening}, respectively. We further show that the proposed bound can reduce to the VC-dimension bound correctly. We then illustrate the benefits of this new approach in the context of the two examples. For the problem of estimating the Gaussian mean mentioned above, we can obtain a bound that is order-wise stronger than previously given in the literature. For the phase retrieval problem considered in \cite{asadi2018chaining}, the bound can be naturally improved by optimizing over a continuous parameter. 

\subsection{Related works}

The approach of using information-theoretic tools to develop generalization error bounds was pioneered by Russo and Zou \cite{russo2016controlling}, and then extended by Xu and Rakinsky \cite{xu2017information}. The bound in \cite{xu2017information} was later tightened by Bu et al. in \cite{bu2020tightening}, using the individual sample mutual information instead of sample set mutual information. Steinke and Zakynthinou \cite{steinke2020reasoning} proposed a conditional mutual information based bound by introducing pseudo samples and split them into training and testing samples. Combining the idea of error decomposition \cite{bu2020tightening} and the conditional mutual information bound in \cite{steinke2020reasoning}, Haghifam et al. \cite{haghifam2020sharpened} provided a sharpened bound based on conditional individual mutual information. Hafez-Kolahi et al. proposed a streamlined view of several bounding techniques proposed in the literature \cite{hafez2020conditioning}. Rodriguez-Galvez et al.\cite{rodriguez2021random} and Zhou  et al.  \cite{zhou2021individually} provide further improved on conditional individual mutual information type of bound. Hellstrom and Durisi \cite{hellstrom2020generalization,hellstrom2020generalizationB} used information density to unify several existing bounds and bounding approaches. Similar bounds using other measures can be found in \cite{rodriguez2021tighter}. The information-theoretic bounds have been used to bound generalization errors in noisy and iterative learn algorithms \cite{pensia2018generalization,li2019generalization,negrea2019information,haghifam2020sharpened,rodriguez2021random}.

\subsection{Organization of the paper}

The rest of the paper is organized as follows. In Section \ref{sec:notation}, we introduce the necessary notation and general background on the chaining technique and successive refinement source coding. The main results are given in Section \ref{sec:main} with a few important discussions. In Section \ref{sec:examples}, we return to the two motivating examples to illustrate the benefits of the proposed bound, 
and Section \ref{sec:conclusion} concludes the paper. 

\section{Notation and preliminaries}
\label{sec:notation}

\subsection{Generalization error}

Consider the supervised learning setting, and denote the data domain as $\Zc := \Xc \times \Yc$, where $\Xc$ is the instance domain and $\Yc$ is the label domain. The hypothesis class is denoted as $\Hc_\Wc = \{ h_{W} : W \in \Wc \} \subseteq \Yc^{\Xc}$, where $\Wc$ is the parameter space, or more generally the index set of the hypothesis class. A learning algorithm has access to a sequence of training samples  $Z_{[n]} = (Z_1, Z_2, \ldots, Z_n)$, where each $Z_i$ is drawn independently from $\Zc$ following some unknown probability distribution $\xi$, and the notation $[n]$ is used to denote the set $\{1,2,\ldots,n\}$. The mapping from the data set $Z_{[n]}$ to the hypothesis can be represented by the kernel $P_{W | Z_{[n]}}$. 

Under a loss function $\ell:  \Wc  \times \Zc \rightarrow \Rb$, the population risk is given as
\begin{align}
L_{\xi}(w) := \Eb_{Z \sim \xi}\ell(w, Z). \label{eqn:def-pop-loss}
\end{align}
For training using the data set $Z_{[n]}$, the empirical risk of a given hypothesis $h_w$ is 
\begin{align}
L_{Z_{[n]}}(w) := \frac{1}{n} \sum_{i = 1}^n \ell(w, Z_i). \label{eqn:def-emp-loss}
\end{align}
The generalization error for the given data set is defined as
\begin{align}
\gen_{Z_{[n]}}(\xi,w) := L_{\xi}(w) - L_{Z_{[n]}}(w). \label{eqn:def-gen}
\end{align}
The expected generalization error is defined as 
\begin{align}
\gen(\xi, P_{W|Z_{[n]}}) := \Eb \left[L_{\xi}(W) - L_{Z_{[n]}}(W) \right], \label{eqn:def-gen-expected}
\end{align}
where the expectation is taken over the joint distribution $P(W,Z_{[n]})=\xi^n \otimes P_{W|Z_{[n]}}$. Our goal is to bound $\gen(\xi, P_{W|Z_{[n]}})$ in this work. 

Generalization error can also be written in a different form by defining 
\begin{align}
\gen^i_{Z_i}(\xi,w) & :=L_{\xi}(w) - \ell(w, Z_i), \\
\gen^i(\xi,P_{W|Z_i}) & :=\Eb[L_{\xi}(W) - \ell(W, Z_i)]. \label{eqn:def-gen-individual}
\end{align}
Clearly $\gen(\xi, P_{W|Z_{[n]}})=\frac{1}{n}\sum_{i=1}^n \gen^i(\xi,P_{W|Z_i})$. It is worth noting that the distribution $P_{W|Z_i}$ is obtained by marginalizing over $P(W,Z_{[n]})$ (and dividing $\xi$).

\subsection{Random process and partitions}

Let $\{X_t\}_{t\in\mathcal{T}}$ be a random process with the index set $\mathcal{T}$. There is a metric $d$ on $\mathcal{T}$ which describes the dependence among the random variables in the random process. For simplicity, $\{X_t\}_{t\in\mathcal{T}}$ is written as $X_\mathcal{T}$ whenever it does not cause confusion.

The following definitions are standard.

\begin{definition}[Separable process]
The random process $X_\mathcal{T}$ on the metric space $(\mathcal{T},d)$ is called separable if there is a dense countable set $\Tc_0 \subset \Tc$ such that for any $t \in \Tc$, there exists a sequence $(s_n)$ in $\Tc_0$ such that $s_n \rightarrow t$ and $X_{s_n} \rightarrow X_{t}~a.s.$.
\end{definition}
All the random processes considered in this paper are separable.

\begin{definition}[Sub-Gaussian process]
The random process $X_\mathcal{T}$ on the metric space $(\mathcal{T},d)$ is called sub-Gaussian if $\Eb [X_t]=0$ for all $t\in\mathcal{T}$, and 
\begin{align}
\Eb \left[ e^{\lambda(X_t-X_s)}\right]\leq e^{\frac{1}{2}\lambda^2d^2(t,s)}, \quad for\,\, all \,\, t,s\in \mathcal{T},\,\lambda\geq 0. \notag
\end{align}
\end{definition}

Let $W$ be a random variable on the set $\mathcal{W}$. From here on, we use capital letter (such as $W$) to denote a random variable, and its lower case letter (such as $w$) to indicate a realization. In the particular setting of generalization bound, the random variable $W$ is the index (or the parameters) of the hypothesis chosen by the possibly randomized learning algorithm using the stochastically generated data set $Z_{[n]}$. The random process is $X_{\Tc} = \{\gen_{Z_{[n]}}(\xi, w)\}_{w\in\mathcal{W}}$ index by $w$ in this case.  $W$ is jointly distributed with $X_{\Wc}$, and $X_W = \gen_{Z_{[n]}}(\xi, W)$ is the generalization error of interest. 

A well known tool in bounding a random process is the chaining technique \cite{talagrand2006generic}. The notion of an increasing sequence of $\epsilon$-partition of the metric space $(\mathcal{T},d)$ is particularly important in this setting. 

\begin{definition}[Increasing sequence of $\epsilon$-partition] A partition $\mathcal{P}=\{A_1,A_2,\ldots,A_m\}$ of the set $\mathcal{T}$ is called an $\epsilon$-partition of the metric space $(\mathcal{T},d)$ if for all $i=1,2,\ldots,m$, $A_i$ can be contained within a ball of radius $\epsilon$. A sequence of partitions $\{\mathcal{P}_k\}_{k=m}^\infty$ of a set $\mathcal{T}$ is called an increasing sequence if for any $k\geq m$ and each $A\in\mathcal{P}_{k+1}$, there exists $B\in\mathcal{P}_k$ such that $A \subseteq B$. 
\end{definition}

In the context of bounding the generalization error, when it is viewed as a random process $X_\mathcal{T}$, we are interested in the expectation $\Eb [X_W]$. 

\subsection{Information theory and successive refinement source coding}
For a discrete random variable $X$, the entropy is denoted as $H(X)$, and for a continuous random variable $Y$, its differential entropy is denoted as $h(Y)$. The mutual information between two random variables $X,Y$ is denoted as $I(X;Y)$, regardless whether they are discrete or continuous. We use natural logarithm in this work, and thus information is measured in terms of nats. 

Successive refinement source coding considers the problem of encoding a source $X$ in a total of $K$ stages, each with rate $R_k$ nats of coding budget, and the end user uses the encoded information in all previous stages, i.e., stages $1,2,\ldots,k$, to reconstruct the source at stage $k$. The achievable rate region, i.e., the set of encoding rate vectors, is given as the collection of nonnegative rates $(R_1,R_2,\ldots,R_K)$ such that 
\begin{align}
\sum_{i=1}^kR_i\geq I(X;{X}_1,{X}_2,\ldots,{X}_{k}), 
\end{align}
where ${X}_k$ is a random variable representing the stochastic reconstruction of $X$ at stage-$k$, which guarantees $\Eb [d(X,{X}_k)]\leq D_k$, i.e., the distortion is less than or equal to the given distortion budget $D_k$; see \cite{rimoldi1994successive} for more details. One particular useful choice is to make $X\leftrightarrow {X}_K\leftrightarrow {X}_{K-1}\leftrightarrow \cdots\leftrightarrow X_1$ a Markov chain. The simplification is immediate since 
\begin{align}
I(X;{X}_1,{X}_2,\ldots,{X}_{k})=I(X;{X}_{k})
\end{align}
in this case. When the choice of $X_K,X_{K-1},\ldots,X_1$ satisfying the Markov chain also yields the optimal coding rates among all possible choices of auxiliary $X_K,X_{K-1},\ldots,X_1$ random variables, the source $X$  is called successively refinable \cite{equitz1991successive}. More results on this problem can be found in \cite{tuncel2003additive,lastras2001all}.

\vspace{-0.2cm}
\section{Main results}
\label{sec:main}

\subsection{Main Theorems}

We define a new notion of stochastic chain as follows. 
\begin{definition}[Stochastic chain of random process and random variable pair] Let $(X_\mathcal{T},W)$ be a random process and random variable pair, where $W$ is a random variable in the set $\mathcal{T}$. A sequence of random variables $\{W_k\}^{\infty}_{k=k_0}$, each distributed in the set $\mathcal{T}$, is called a stochastic chain of the pair $(X_\mathcal{T},W)$,  if 1) $\lim_{k\rightarrow \infty} \Eb[X_{W_{k}}]=\Eb[X_W]$, 2) $\Eb [X_{W_{k_0}}]=0$, and 3) $\{X_t\}_{t\in\mathcal{T}}\leftrightarrow W\leftrightarrow W_{k}\leftrightarrow W_{k-1}$ is a Markov chain for every $k> k_0$. 
\end{definition}

We allow $k_0$ to take the value of $-\infty$ instead of providing another parallel definition to that effect. We are now ready to present the first main theorem of this work. 

\begin{theorem}
\label{theorem:firstgen}
Assume $\{\gen_{Z_{[n]}}(\xi,w)\}_{w\in\mathcal{W}}$ is sub-Gaussian on $(\mathcal{W},d)$, and $\{W_k\}^{\infty}_{k=k_0}$  is a stochastic chain of $(\{\gen_{Z_{[n]}}(\xi,w)\}_{w\in\mathcal{W}},W)$. 
 Then 
\begin{align}
 & \gen(\xi,P_{W|Z_{[n]}}) \notag \\
&\leq \sum_{k=k_0+1}^\infty\Eb\left[d(W_k,W_{k-1})\sqrt{2D(P_{Z_{[n]}|W_k}||P_{Z_{[n]}})}\right].
\end{align}
Moreover, we have
\begin{align}
    & \gen(\xi,P_{W|Z_{[n]}}) \notag \\
    & \leq  \sum_{k=k_0+1}^\infty \sqrt{\Eb[d^2(W_k,W_{k-1})]}\sqrt{2I(Z_{[n]}; W_k)}.
\end{align}
\end{theorem}

The following theorem is based on the individual sample mutual information bound of \cite{bu2020tightening}. 

\begin{theorem}
\label{theorem:secondgen}
For each $i\in [n]$, assume $\{\gen_{Z_i}^i(w)\}_{w\in\mathcal{W}}$ is sub-Gaussian on $(\mathcal{W},d)$, and $\{W_{i,k}\}^{\infty}_{k=k_0}$ is a stochastic chain of $(\{\gen_{Z_i}^i(w)\}_{w \in \mathcal{W}},W)$.  
Then
\begin{align}
& \gen(\xi,P_{W|Z_{[n]}}) \notag \\
&\leq \frac{1}{n}\sum_{i=1}^n\sum_{k=k_0+1}^\infty\Eb\left[d(W_{i,k},W_{i,k-1})\right.\notag\\
&\qquad\qquad\qquad\qquad \left.\cdot\sqrt{2D(P_{Z_{i}|W_{i,k}}||P_{Z_{i}})}\right].
\end{align}
Moreover, we have
\begin{align}
    & \gen(\xi,P_{W|Z_{[n]}}) \notag \\
    & \leq  \frac{1}{n}\sum_{i=1}^n \sum_{k=k_0+1}^\infty \sqrt{\Eb[d^2(W_{i,k},W_{i,k-1})]}\sqrt{2I(Z_{i}; W_{i,k})}.
\end{align}
\end{theorem}

These two theorems are given in the context of bounding generalization errors, which are obtained using a more general result on bounding random processes. 
\begin{theorem}
\label{theorem:general}
Assume $X_\mathcal{T}$ is sub-Gaussian on $(\mathcal{T},d)$, and $\{W_k\}^{\infty}_{k=k_0}$ is a stochastic chain for $(X_\mathcal{T},W)$, then
\begin{align}
& \Eb\left[X_W\right] \notag \\
& \leq \sum_{k=k_0+1}^\infty\Eb\left[d(W_k,W_{k-1})\sqrt{2D(P_{X_{\mathcal{T}}|W_k}||P_{X_{\mathcal{T}}})}\right].
\end{align}
Moreover, we have
\begin{align}
    & \Eb\left[X_W\right] \notag \\
    & \leq  \sum_{k=k_0+1}^\infty \sqrt{\Eb[d^2(W_k,W_{k-1})]}\sqrt{2I(X_{\mathcal{T}}; W_k)}.
\end{align}
\end{theorem}

By using a deterministic sequence of partitions to form $\{W_k\}_{k_0}^\infty$, we recover the result in \cite{asadi2018chaining} which was obtained for bounded metric space $(\mathcal{T},d)$.
\begin{corollary} \label{corollary:partition}
Let $\{\mathcal{P}_k\}_{k=k_0}^\infty$ be an increasing sequence of partitions of $\mathcal{T}$, where for each $k\geq k_0$, $\mathcal{P}_k$ is a $2^{-k}$-partition of the bounded metric space $(\mathcal{T},d)$, and $2^{-k_0}\geq \text{diam}(\mathcal{T})$. Let $W_{k}$ be the center of the covering ball of the partition cell that $W$ belongs to in the partition $\mathcal{P}_k$, then for separable process $X_{\Tc}$ on $(\Tc, d)$,
\begin{align}
\Eb\left[X_W\right]&\leq \sum_{k=k_0+1}^\infty\Eb\left[3\cdot2^{-k} \sqrt{2D(P_{X_{\mathcal{T}}|W_k}||P_{X_{\mathcal{T}}})}\right] \notag \\
& \leq \sum_{k=k_0+1}^\infty 3\cdot2^{-k}\sqrt{2I(X_{\mathcal{T}};W_k)}.
\end{align}
\end{corollary}

The proof of Theorem \ref{theorem:general} can be found in the supplementary material. To obtain Theorem \ref{theorem:firstgen} from Theorem \ref{theorem:general}, we let $\mathcal{T}:=\mathcal{W}$, and $X_w:=\gen(w)$ for $w\in\mathcal{W}$. Due to the Markov chain 
\begin{align}
X_{\mathcal{T}}=\{\gen(w)\}_{w\in\mathcal{W}}\leftrightarrow Z_{[n]}\leftrightarrow W\leftrightarrow W_k,
\end{align}
for all $k\geq k_1$, we can apply the data processing inequality for KL divergence \cite{wu2017lecture} and that for mutual information, respectively, to arrive at 
\begin{align}
D(P_{X_{\mathcal{T}}|W_k}||P_{X_{\mathcal{T}}}) & \leq D(P_{Z_{[n]}|W_k}||P_{Z_{[n]}}), \notag\\
I(X_\mathcal{T};W_k) &\leq I(Z_{[n]};W_k),
\end{align}
from which Theorem \ref{theorem:firstgen} follows immediately. Theorem \ref{theorem:secondgen} can be obtained similarly.

When the process is not sub-Gaussian, more general forms of these bounds can also be found in terms of the cumulant generating function. This result is given in the supplementary material. 

\subsection{Relations to existing results}

\paragraph{Connection to VC theory:}  For binary classification problems, i.e., $|\Yc| = 2$ with zero-one loss $\ell(w, (x, y)) = \mathbb{I}(h_w(x) \not= y)$, the generalization error of any classifier $W$ is upper bounded as $\gen(\xi, P_{W|Z_{[n]}}) \leq O(\sqrt{\frac{d_{VC}(\Wc)}{n}})$, where $d_{VC}(\Wc)$ is the VC-dimension of the classification function class $\Hc_\Wc$ {(c.f.,  \cite{shalev2014understanding} Ch. 6)}. The generalization error bound in Theorem \ref{theorem:general}, or more precisely the proposed stochastic chaining approach, can naturally recover the VC-dimension based bound, and we establish this connection in the supplementary material.

\paragraph{Discussion on the chaining construction:} The conventional deterministic chaining places certain structural constraints on the hierarchical partitions. For example, consider a partition of a bounded 2-D space using congruent hexagon cells; the next partition at the higher level will be collections of such hexagons. This subsequently implies that hierarchy must follow a certain relation between consecutive levels, and the analysis of such hierarchical partitions can be complex. The stochastic chaining technique can remove the geometric constraints in the \textit{design of hierarchical partitions} as in Corollary \ref{corollary:partition} in many cases. In the example above, we can replace the partition using either an additive Gaussian noise or additive noise with a uniform distribution on hexagons (see the second example in the next section where a similar uniform additive noise is used). 

Since stochastic chains include conventional partition-based chaining as a special case, it is not more difficult to construct. In fact, the construction can be more straightforward due to its flexibility. For example, for bounded metric space, we can use the following generic construction: let $p(W_{k-1}|W_k)$ be uniformly distributed on a metric ball of radius $2^{-k}$ centering at $W_k$. If more information regarding the distribution of $W$ is known, we can further optimize the chain, e.g., by adjusting the radius such that they are dependent on the density value of $W_k$; more specifically, we can let the radius be larger for $W_k$ values of lower density, and vice versa. If the metric space is also a vector space, it can be convenient to let $p(W_{k-1}|W_k)$ be some vector Gaussian distribution with covariance scaling like $2^{-k}$. This allows more opportunity for optimization for stronger bounds in a parametric form. In contrast, it is impossible to design partitions (or deterministic mappings \cite{hafez2020conditioning}) to mimic such behaviors, let alone finding analytic bound. This issue in fact has a natural origin in source coding: deterministic quantization design vs. probabilistic forward test channel modeling. The latter is used in source coding for mathematically precise characterization, and for analytic optimization.

\paragraph{{Comparison to the chaining technique in \cite{hafez2020conditioning}:}} The alternative chaining method proposed by Hafez-Kolahi et al. (Theorem 6 in \cite{hafez2020conditioning}) used a different chaining construction, which does not require hierarchical partitions, and to some extent it helps resolve the difficulty in designing such hierarchical partitions. However, this simplification came with a heavy price: the learning algorithm must be \textit{deterministic}, the hypothesis space $\mathcal{W}$ still needs to be \textit{bounded} (since the core steps rely on \cite{asadi2018chaining}), and there is a factor of 2 loss in the bound. The restrictions make it inapplicable in the two examples we study in the next section. In contrast, the proposed method applies to unbounded metric space, and does not require the learning algorithm to be deterministic.

\section{Two examples: estimating the Gaussian mean and phase retrieval}
\label{sec:examples}

We analyze two simple settings, which demonstrate the effectiveness of the proposed stochastic chaining technique. The purpose for discussing the following two examples is by no means to literally characterize the generalization error, since the generalization error can be calculated directly due to the simplicity of the examples. We aim to show the effectiveness of the proposed stochastic chaining technique in these two  examples by comparing with the underlining generalization error and some previous generalization error bound.

\subsection{Estimating the Gaussian Mean}
Consider the case when the training samples $Z_{[n]}$ are drawn $i.i.d.$ following $N(\mu, \sigma^2)$ for some unknown $\mu$. Here $\mathcal{T}=\mathbb{R}$, and a natural choice of the metric on this space is the (scaled) Euclidean distance. The loss function is $\ell(w,Z)=(w-Z)^2$, and by defining $\bar{Z}_n:=\frac{1}{n}\sum_{i=1}^nZ_i$, the random process (indexed by $w$) of interest can be written as
\begin{align}
\gen_{Z_{[n]}}(\xi,w)=\sigma^2 + \mu^2 - \frac{\sum_{i=1}^n Z_i^2}{n} + 2w (\bar{Z}_n-\mu).
\end{align}
It follows that 
\begin{align}
\gen_{Z_{[n]}}(\xi,w)- \gen_{Z_{[n]}}(\xi,v)= 2(w - v) \left(\bar{Z}_n - \mu \right),
\end{align}
which is $d^2(w, v)$ sub-Gaussian with $d^2(w, v) =  \frac{4 \sigma^2 (w - v)^2}{n}$.
The learner deterministically estimates $\mu$ by averaging the training samples, i.e., $W =\bar{Z}_n$. We shall use Theorem \ref{theorem:firstgen} to bound the generalization error in this case.

To build a stochastic chain, select a sequence of mutually independent Gaussian noise $\left\{ N'_{i} \right\}_{i \in \Nb}$, which is independent of $W$, and $N'_i \sim \Nc(0, \sigma'^2_i)$, where $\sigma_{i}'^2 = \frac{\sigma^2}{2^i n}$. Define the cumulative noise 
\begin{align}
\vspace{-0.1cm}
N_{k} := \sum_{i = k + 1}^{\infty} N'_i \sim \Nc(0, \sigma^2_k),
\end{align}
where $ \sigma^2_k = \frac{\sigma^2}{2^k n}$.
The stochastic chain is designed as
\begin{align}
W_{k} - \mu & = \alpha_k(W - \mu + N_k),
\end{align}
where $\alpha_k =\frac{\sigma^2/n}{\sigma^2/n + \sigma_k^2} = \frac{1}{1 + 2^{-k}}$.
We then have
\begin{align}
W_{k-1} - \mu = \frac{\alpha_{k-1}}{\alpha_k} (W_k - \mu) + \alpha_{k-1}N'_{k}, \label{eqn:W_k-1_k}
\end{align}
where $W_k$ and $N'_{k}$ are independent. 
Under this stochastic chain, we can derive the expression for $\sqrt{\Eb[d(W_k, W_{k-1})^2]}$ 
and the mutual information term $I(Z_{[n]}; W_k)$. Specifically, $\Eb[d(W_k, W_{k-1})^2] \leq \frac{\sigma^4}{n^2} \frac{3}{2^{k-1} + 1}$, which relies on the relations between $W_k$ and $W_{k-1}$ in (\ref{eqn:W_k-1_k}) and the detailed calculation is given in the supplementary material. The mutual information can be upper bounded as
\begin{align}
I(Z_{[n]}; W_k) \leq I(W; W_{k}) = \frac{1}{2} \ln(1 + 2^k),
\end{align}
where the inequality is due to the data processing inequality over the Markov chain $Z_{[n]} \leftrightarrow W \leftrightarrow W_k$ and the equality is by the Gaussian channel nature of the stochastic chain design. The detailed proof steps
are given in the supplementary material. A bound of the following form can then be obtained 
\begin{align}
 & \Eb[X_W]  \leq \frac{\sigma^2}{n} \sum_{k = -\infty}^{\infty} \sqrt{\frac{3 \ln(1 + 2^k)}{2^{k-1} + 1}}. \label{eqn:sums}
\end{align}
Note that the series sum on the right hand side of (\ref{eqn:sums}) converges, and thus the bound is of order $O( \sigma^2/n)$. Bounding the series sum using numerical methods, we can then obtain $\Eb[X_W] \leq  \frac{13\sigma^2}{n}$. 

Due to the simplicity of the setting, the generalization error can in fact be calculated exactly to be $\frac{2 \sigma^2}{n}$. It can be seen that the generalization bound offered by Theorem \ref{theorem:firstgen} has the same $O(\sigma^2/n)$ order as the true generalization error. In contrast, the authors of \cite{bu2020tightening} derived a generalization error bound of the order $O(\sigma^2/\sqrt{n})$ using the individual sample mutual information approach. Thus the proposed approach results in an order-wise improvement in this example case. More importantly, it can be seen that the proposed chaining approach allows us to overcome the limitation of bounded metric space (i.e., the chaining mutual information approach \cite{asadi2018chaining} does not even apply in this setting), and also simplify the calculation due to the introduced dependence structure in the chain. In the supplementary material, we further derive an improved bound (with a slightly better constant factor) using Theorem \ref{theorem:secondgen}. 

\subsection{Phase retrieval}

\begin{table*}
  \caption{Comparison of $\Eb [X_W]$ bounds}
  \label{table:EX_W}
  \centering
  \setlength\tabcolsep{4pt}
  \begin{tabular}{cccccccc}
    \toprule
    $\epsilon$                                      & $1/20$    &  $1/30$  &  $1/40$   &   $1/50$   &   $1/100$ &   $1/200$   &   $1/400$\\
    \midrule    \midrule
    Chaining mutual information\cite{asadi2018chaining} &  1.1013     & 0.7507  & 0.5709 & 0.4612 & 0.2364 & 0.1204 & 0.0610 \\
     \midrule
    stochastic chaining ($\gamma=3.75$)                          &  0.4951     &0.3387  &  0.2581 & 0.2088 & 0.1074 & 0.0548 & 0.0278 \\
     \midrule
    $\Eb [X_W]$ true value                    & 0.0626      & 0.0417  & 0.0313 & 0.0250 & 0.0125& 0.0062  & 0.0031  \\
    \bottomrule
  \end{tabular}
\end{table*}

In the phase retrieval example given in \cite{asadi2018chaining}, the data $Z:=(G_1,G_2)\sim \mathcal{N}(0,I_2)$ is a standard normal vector in $\mathbb{R}^2$. The hypothesis class is $\mathcal{W}=[0,2\pi)$, and through the transformation $t=(\cos w,\sin w)$ for $w\in \mathcal{W}$, it is in fact the same as $\mathcal{T}=\{t\in R^2: ||t||_2=1\}$; we will use them interchangeably. Define the loss function $\ell(t,Z)=-\langle t, Z \rangle$, which implies that the learner wishes to estimate an angle for the underlying data, and the generalization error process is a Gaussian process $X_t:=\langle t, Z \rangle$. The metric $d$ is the Euclidean distance, and the process $X_\mathcal{T}$ is sub-Gaussian. Suppose the learned parameter is
\begin{align}
W:=\left(\argmax_{\phi\in [0,2\pi)} X_{\phi}\right)\oplus \zeta \, (\text{mod } 2\pi),
\end{align}
where $\zeta$ is independent of $X_\mathcal{T}$, and has an atom with a mass $\epsilon$ on 0, and $1-\epsilon$ that is uniformly distributed in $[0,2\pi)$. Note that $\argmax_{\phi\in [0,2\pi)} X_{\phi}$ is exactly the phase of $(G_1,G_2)$, which will be the hypothesis learned by an ERM learner, and $W$ being retrieved here is a noisy version of the phase. 

The stochastic chain can be given as 
\begin{align}
W_k=(W\oplus N_k)(\text{mod } 2\pi), 
\end{align}
where $N_k=\sum_{i=k+1}^\infty N'_i$, and $N'_k$ is uniformly distributed on $[-\gamma^{-k}\pi,\gamma^{-k}\pi)$ for some $\gamma>1$ to be specified later; $N'_k$'s are mutually independent and also independent of the hypothesis parameter $W$. 

Since $W\oplus N_{-1}$ is independent of $Z$ and uniformly distributed on $[0,2\pi)$, we have $\Eb[X_{W_{-1}}]=\Eb[\langle W+N_{-1}, Z \rangle]=0$. It is also clear that $W_k\rightarrow W$ when $k\rightarrow \infty$ a.s., and thus $\Eb [X_{W}]=\lim_{k\rightarrow \infty}\Eb [X_{W_k}]$ since the process is Gaussian. Since $W_{k-1}-W_{k}$ is exactly $N'_{k}$, the Euclidean distance between $W_k$ and $W_{k-1}$ (using their vector representations) is bounded by the length of the arc, i.e., $d(W_{k},W_{k-1})\leq \gamma^{-k}\pi$. 
We can now apply Theorem \ref{theorem:firstgen}, where 
\begin{align}
I(W_{k};X_\mathcal{T}) & =h(N_{k}\oplus W)-h(N_{k}\oplus \zeta) \notag \\
& = \log 2\pi-h(N_{k}\oplus \zeta).
\end{align}
The second term can be bounded as 
\begin{align}
h(N_{k}\oplus \zeta)& \geq h\left(N_{k}\oplus \zeta\bigg{|}\sum_{k+2}^{\infty}N'_j\right) =h(N'_{k+1}\oplus \zeta), \label{eqn:relaxeddiffent}
\end{align}
using the fact that more conditioning reduces the differential entropy.
Due to the structure of the distribution of $N'_{k+1}$ and $Z$, the density of $N'_{k+1}\oplus \zeta$ can be written down explicitly as
\begin{align}
f(N'_{k+1}+\zeta)=\left \{
\begin{array}{ll}
(2\pi)^{-1}(1-\epsilon) \\
\quad \quad \left[-\pi, -\gamma^{-k-1}\pi\right)\cup\left[ \gamma^{-k-1}\pi,\pi\right)\\
(2\pi)^{-1}(\gamma^{k+1}\epsilon+(1-\epsilon))\\
\quad \quad  \left[-\gamma^{-k-1}\pi,\gamma^{-k-1}\pi\right).
\end{array}
\right.
\end{align}
Thus we can bound $h(N_{k}\oplus \zeta)$ and subsequently $I(W_{k};X_T)$ using this density function, which eventually gives
\begin{align}
\Eb[X_W]&\leq \sqrt{2}\pi\sum_{k=0}^{\infty}\gamma^{-k}{\Big (}(1-\epsilon)(1-\frac{1}{\gamma^{k+1}})\log (1-\epsilon) \notag \\
& \quad+\left[\epsilon+\frac{1-\epsilon}{\gamma^{k+1}}\right]\log\left[\gamma^{k+1}\epsilon +1-\epsilon\right] {\Big )}^{1/2}. \label{eqn:phasefinal}
\end{align}

When choosing $\gamma=2$, this is almost identical to the result given in \cite{asadi2018chaining} using the partition based chaining, except the slightly better coefficient $\sqrt{2}\pi$ instead of $6\sqrt{2}$. This improved coefficient is mainly due to the more explicit bound on $d(W_{k},W_{k-1})$ inherent in the Euclidean space, instead of the same distance derived in a generic metric space. 

One advantage of the proposed approach is that we can further optimize $\gamma$ over $\mathbb{R}$. Observe that the series has a faster decaying tail if $\gamma$ is large, however, the first term, i.e., $k=0$, approaches $\infty$ when $\gamma\rightarrow \infty$. Thus there is an optimal $\gamma$ value in between for this bound. Numerical result suggests $\gamma^*\approx3.75$, which provides a slight improvement comparing to $\gamma=2$. As noted in \cite{asadi2018chaining}, in this toy setting, we can in fact calculate the exact true value $\Eb [X_W]=\epsilon\frac{\sqrt{\pi}}{2}$. A comparison of several bounds is given in Table. \ref{table:EX_W}. To obtain (\ref{eqn:phasefinal}), we have in fact relaxed this bound in (\ref{eqn:relaxeddiffent}) for convenience using a simple property of the entropy function, and therefore loosen the bound to some extent. Moreover, we have chosen to use the geometric sequence $\gamma^k$ to produce the stochastic chain, and it is possible other sequences can produce tighter bounds. 

The individual sample mutual information bound in \cite{bu2020tightening} requires multiple samples. In this phase retrieval example, however, there is only one sample $G^2$, and this bound degrades to the mutual information based bound in \cite{xu2017information}, which in this case is vacuous since $I(W; X_{\Tc})$ is in fact infinite.

\section{Conclusion}
\label{sec:conclusion}

We proposed a new chaining-based approach to bound the generalization error by replacing the hierarchical partitions with a stochastic chain.  The proposed approach can firstly remove naturally the restriction for the metric space to be bounded, and secondly, it helps to simplify the computation, and lastly, it can remove the geometric constraints in designing the hierarchical partitions in some cases. Two examples are used to illustrate that the proposed approach can overcome some difficulties in applying the chaining mutual information approach. The roles that chaining can play in bounding generalization error on conjunction with other information-theoretic approach, such as the conditional mutual information \cite{steinke2020reasoning}, information density \cite{hellstrom2020generalizationB}, and Wasserstein distance \cite{rodriguez2021tighter}, as well as the possible application in noisy and stochastic learning algorithms, call for further research. 

\bibliographystyle{IEEEtran}

\onecolumn
\appendix
\renewcommand\thesubsection{\Alph{\thesection}.\roman{subsection}}

\section{Proof of Theorem 3.3}
\label{sec:outline}
To prove the theorem, we start by writing
\begin{align}
X_W=X_{W_{k_0}}+\sum_{k=k_0+1}^{k_1}(X_{W_{k}}-X_{W_{k-1}})+(X_W-X_{W_{k_1}}).
\end{align}
Because $\{W_k\}^{\infty}_{k=k_0}$ is a stochastic chain for $(\tilde{X}_\mathcal{T},W)$, we have $\Eb [X_{W_{k_0}}]=0$ and $\lim_{k_1 \rightarrow \infty} \Eb[X_{W_{k_1}}] = \Eb[X_{W}]$, and it follows that
\begin{align}
\Eb\left[X_W\right]
&=\sum_{k=k_0+1}^{\infty}\Eb\left[X_{W_{k}}-X_{W_{k-1}}\right]\notag\\
&=\sum_{k=k_0+1}^{\infty}\Eb\left[\Eb[X_{W_{k}}-X_{W_{k-1}}|W_k,W_{k-1}]\right]. \label{eqn:chain-decompose}
\end{align}

By the Donsker–Varadhan variational representation of the KL divergence, the expectation of a function $g(Y)$ with respect to the measure $P$ defined on $\mathcal{Y}$ can be bounded as
\begin{align}
\Eb_{P}[g(Y)]\leq \inf_{\lambda>0} \frac{1}{\lambda}\left(D(P||Q)+\log \Eb_{Q}[e^{\lambda g(Y)}]\right),
\end{align}
where $Q$ is another measure on $\mathcal{Y}$.

In our setting, let $Y=g(Y)=\Delta X_{w_k,w_{k-1}}$, $P=P_{\Delta X_{w_k,w_{k-1}}|w_k,w_{k-1}}$, and $Q=P_{\Delta X_{w_k,w_{k-1}}}$, then we have 
\begin{align}
&\Eb_{P_{\Delta X_{W_{k},W_{k-1}}|w_k,w_{k-1}}}[\Delta X_{w_k,w_{k-1}}]\notag\\
&\leq \inf_{\lambda>0}\frac{1}{\lambda}{\Big (}D(P_{\Delta X_{w_k,w_{k-1}}|w_k,w_{k-1}}||P_{\Delta X_{w_k,w_{k-1}}}) +\log \Eb_{P_{\Delta X_{w_k,w_{k-1}}}}\left[e^{\lambda (X_{w_k}-X_{w_{k-1}})}\right] {\Big )}\notag\\
&\leq \inf_{\lambda>0}\frac{1}{\lambda}{\Big (}D(P_{\Delta X_{w_k,w_{k-1}}|w_k,w_{k-1}}||P_{\Delta X_{w_k,w_{k-1}}})  +\frac{1}{2}d^2(w_k,w_{k-1})\lambda^2 {\Big )}\notag\\
&=d(w_k,w_{k-1})\sqrt{2D(P_{\Delta X_{w_k,w_{k-1}}|w_k,w_{k-1}}||P_{\Delta X_{w_k,w_{k-1}}})},
\end{align}
where the second inequality is because the process $X_{\mathcal{T}}$ is sub-Gaussian on $(\mathcal{T},d)$.

The fact that $\{W_k\}^{\infty}_{k=k_0}$ is a stochastic chain also implies that $\lim_{k\rightarrow \infty} \Eb[X_{W_{k}}]=\Eb[X_W]$, and thus 
\begin{align}
& \Eb\left[X_W\right] \leq \sum_{k=k_0+1}^\infty\Eb\left[d(W_k,W_{k-1})\sqrt{2D(P_{\Delta_k|W_k,W_{k-1}}||P_{\Delta_k})}\right].
\end{align}
Denote $\tilde{X}_{\Tc}$ as an independent copy of $X_{\Tc}$ such that $\tilde{X}_{\Tc}$ and $X_{\Tc}$ are independent and have the same distribution. $P_{\Delta_k}$ is the distribution of $\Delta \tilde{X}_{W_k, W_{k-1}}$ conditioned on $W_{k}, W_{k-1}$. 
By the data processing inequality for the KL divergence, we have
\begin{align}
D(P_{\Delta_k|W_k,W_{k-1}}||P_{\Delta_k})\leq D(P_{X_{\mathcal{T}}|W_k,W_{k-1}}||P_{X_{\mathcal{T}}}),
\end{align}
from which the second inequality follows. 

Let us now consider the mutual information based bound. It is seen that
\begin{align}
& \Eb\left[\Eb[X_{W_{k}}-X_{W_{k-1}}|W_k,W_{k-1}]\right] \notag \\
& = \int_{\Tc^2} \Eb_{P_{X_{\Tc}|w_k,w_{k-1}}}[\Delta X_{w_k,w_{k-1}}] dP_{W_k, W_{k-1}}(w_k, w_{k-1}) \notag \\
&\leq \int_{\Tc^2}  \inf_{\lambda>0}\frac{1}{\lambda}{\Big (}D(P_{X_{\Tc}|w_k,w_{k-1}}||P_{X_{\Tc}})  +\frac{1}{2}d^2(w_k,w_{k-1})\lambda^2 {\Big )} dP_{W_k, W_{k-1}}(w_k, w_{k-1}) \notag\\
& \leq \inf_{\lambda>0}\frac{1}{\lambda}{\Big (}\int_{\Tc^2} D(P_{X_{\Tc}|w_k,w_{k-1}}||P_{X_{\Tc}}) dP_{W_k, W_{k-1}}(w_k, w_{k-1})   + \int_{\Tc^2}  \frac{1}{2}d^2(w_k,w_{k-1})\lambda^2 dP_{W_k, W_{k-1}}(w_k, w_{k-1}){\Big )}  \notag\\
&= \inf_{\lambda > 0} \frac{I(X_{\Tc} ; W_k, W_{k-1})}{\lambda} + \frac{\lambda}{2} \Eb[d^2(W_k, W_{k-1})] \notag \\
& = \sqrt{\Eb[d^2(W_k, W_{k-1})]} \sqrt{2 I(X_{\mathcal{T}} ; W_k)}.
\end{align}
Combing with (\ref{eqn:chain-decompose}) we arrive at
\begin{align}
 \Eb\left[X_W\right] & \leq \sum_{k=k_0+1}^\infty\sqrt{\Eb[d^2(W_k, W_{k-1})]} \sqrt{2 I(X_{\Tc} ; W_k)} ,
\end{align}
which concludes the proof.
\qed

\section{A chaining bound in a more general form}

In this section we provide a more general bound without the assumption on sub-Gaussianity, and replace it with a more general form on the measure concentration.

\begin{definition}
Let $X$ be a real-valued random variable. The cumulant generating function of $X$ is $\Lambda_X(\lambda):=\log \Eb[e^{\lambda X)}]$ for $\lambda\in \mathbb{R}$. 
\end{definition}
If $\Lambda_X(\lambda)$ exists, then $\Lambda_X(0)=0$ and $\Lambda_X'(0)=\Eb X$, and it is convex.  
\begin{definition}

For a convex function $\psi$ defined on the interval $[0,b)$, where $0<b\leq \infty$, its Legendre dual $\psi^*$ is defined as
\begin{align}
\psi^*(x):=\sup_{\lambda \in [0,b)} (\lambda x-\psi(\lambda)).
\end{align}

\begin{lemma}
Assume that $\psi(0)=\psi'(0)=0$, then $\psi^*(x)$ is a non-negative convex and non-decreasing function on $[0,\infty)$ with $\psi^*(0)=0$. Moreover, its inverse function $\psi^{*-1}(y)=\inf\{x\geq 0: \psi^*(x)\geq y\}$ is concave, and can be written as 
\begin{align}
\psi^{*-1}(y)=\inf_{\lambda \in (0,b)}\left(\frac{y+\psi(\lambda)}{\lambda}\right).
\end{align}
\end{lemma}
\end{definition}

\begin{theorem}
Assume $X_\mathcal{T}$ is a random process defined on the metric space $(\mathcal{T},d)$, with $\Eb [X_t]=0$ for all $t\in\mathcal{T}$ and
\begin{align}
\log \Eb\left[ e^{\frac{\lambda (X_t-X_s)}{d(t,s)}}\right]\leq \psi(\lambda), \,\, \text{for all}\,\, t,s\in \mathcal{T},\,\lambda\geq 0, \label{eqn:tail}
\end{align}
where $\psi$ is convex and $\psi(0)=\psi'(0)=0$. Let $\{W_k\}^{\infty}_{k=k_0}$ be a stochastic chain for the random process and random variable pair $(X_\mathcal{T},W)$, then
\begin{align}
\Eb\left[X_W\right] 
&\leq \sum_{k=k_0+1}^\infty\Eb\left[d(W_k,W_{k-1})\psi^{*-1}\left(D(P_{X_{\mathcal{T}}|w_k}||P_{X_{\mathcal{T}}})\right)\right]\label{eqn:generalbound}
\end{align}
Particularly, if $d(W_{k},W_{k-1})\leq \sigma_k$, then 
\begin{align}
\Eb\left[X_W\right]&\leq \sum_{k=k_0+1}^\infty\sigma_k\Eb\left[\psi^{*-1}\left(D(P_{X_{\mathcal{T}}|w_k}||P_{X_{\mathcal{T}}})\right)\right]\leq \sum_{k=k_0+1}^\infty\sigma_k\psi^{*-1}\left(I(X_{\mathcal{T}};W_k)\right).
\end{align}
\end{theorem}

\begin{proof}
To prove the theorem, we start by writing
\begin{align}
X_W=X_{W_{k_0}}+\sum_{k=k_0+1}^{k_1}(X_{W_{k}}-X_{W_{k-1}})+(X_W-X_{W_{k_1}}).
\end{align}
Because $\{W_k\}^{\infty}_{k=k_0}$ is a stochastic chain for $(\tilde{X}_\mathcal{T},W)$, we have $\Eb [X_{W_{k_0}}]=0$, and it follows that
\begin{align}
\Eb\left[X_W\right]-\Eb\left[X_W-X_{W_{k_1}}\right]
&=\sum_{k=k_0+1}^{k_1}\Eb\left[X_{W_{k}}-X_{W_{k-1}}\right]\notag\\
&=\sum_{k=k_0+1}^{k_1}\Eb\left[\Eb[X_{W_{k}}-X_{W_{k-1}}|W_k,W_{k-1}]\right],
\end{align}
By the Donsker–Varadhan variational representation of the KL divergence, the expectation of a function $g(X)$ with respect to the measure $P$ defined on $\mathcal{X}$ can be bounded as
\begin{align}
\Eb_{P}[g(X)]\leq \inf_{\lambda>0}\frac{D(P||Q)+\log \Eb_{Q}[e^{\lambda g(x)}]}{\lambda},
\end{align}
where $Q$ is another measure on $\mathcal{X}$.
In our setting, let $P=P_{\Delta_k|w_k,w_{k-1}}$, $Q=P_{\Delta_k}$, and $g(X)=X_{w_{k}}-X_{w_{k-1}}$, then we have 
\begin{align}
\frac{\Eb_{P_{X_{\Tc}|w_k,w_{k-1}}}[X_{w_{k}}-X_{w_{k-1}}]}{d(w_k,w_{k-1})}
&\leq \inf_{\lambda>0}\frac{D(P_{X_{\Tc}|w_k,w_{k-1}}||P_{X_{\Tc}})+\log \Eb_{P_{X_{\Tc}}}\left[e^{\frac{\lambda(X_{w_k}-X_{w_{k-1}})}{d(w_k,d_{k-1})}}\right]}{\lambda}\notag\\
&\leq \inf_{\lambda>0}\frac{D(P_{X_{\Tc}|w_k,w_{k-1}}||P_{X_{\Tc}})+\psi(\lambda)}{\lambda}\notag\\
&= \psi^{*-1}\left(D(P_{X_{\Tc}|w_k,w_{k-1}}||P_{X_{\Tc}})\right),
\end{align}
where the second inequality is by assumption (\ref{eqn:tail}). 
The fact that $\{W_k\}^{\infty}_{k=k_0}$ is a stochastic chain also implies that $\lim_{k\rightarrow \infty} \Eb[X_{W_{k}}]=\Eb[X_W]$, and thus we arrive at 
\begin{align}
\Eb\left[X_W\right]\leq \sum_{k=k_0+1}^\infty\Eb\left[d(W_k,W_{k-1})\psi^{*-1}\left(D(P_{X_{\Tc}|w_k,w_{k-1}}||P_{X_{\Tc}})\right)\right].
\end{align}
Moreover, since $\{W_k\}^{\infty}_{k=k_0}$ is a stochastic chain, the Markov condition implies that
\begin{align}
P_{X_{\mathcal{T}}|W_k,W_{k-1}}=P_{X_{\mathcal{T}}|W_k},
\end{align}
and thus
\begin{align}
D(P_{X_{\Tc}|W_k,W_{k-1}}||P_{X_{\mathcal{T}}})=D(P_{X_{\mathcal{T}}|W_k}||P_{X_{\mathcal{T}}}),
\end{align}
from which, the first inequality in (\ref{eqn:generalbound}) follows.

If $d(W_k,W_{k-1})\leq \sigma_k$, then we have 
\begin{align}
\Eb\left[X_W\right]&\leq \sum_{k=k_0+1}^\infty\Eb\left[d(W_k,W_{k-1})\psi^{*-1}\left(D(P_{X_{\mathcal{T}}|w_k}||P_{X_{\mathcal{T}}})\right)\right]\notag\\
&\leq \sum_{k=k_0+1}^\infty\Eb\left[\sigma_k\psi^{*-1}\left(D(P_{X_{\mathcal{T}}|w_k}||P_{X_{\mathcal{T}}})\right))\right]\notag\\
&\leq \sum_{k=k_0+1}^\infty \sigma_k\psi^{*-1}(I(X_{\mathcal{T}};W_k)),
\end{align}
where the last inequality is due to Jensen's inequality and $\Eb \left[D(P_{X_{\mathcal{T}}|W_k}||P_{X_{\mathcal{T}}})\right]=I(X_{\mathcal{T}};W_k)$ by definition.
\end{proof}

\section{Connection to VC dimension} 

Consider the binary classification problem with zero-one loss with hypothesis class parameterized by $\Wc$ and data $Z_{[n]} \sim \xi^n$. We follow the similar approach as in \cite{hafez2020conditioning}. 
Independent random variables $Z^{+}_{[n]} \sim \xi^n$, $Z^{-}_{[n]} \sim \xi^n$ and $R_{[n]} \sim \text{Rademacher}^n$ determines the random process $\Xc_{\Wc}$ with
\begin{align}
X_{w} = \frac{1}{\sqrt{n}} \sum_{i = 1}^{n} R_i (\ell(w, Z_i^+) - \ell(w, Z_i^{-})), \quad \forall w \in \Wc.
\end{align}
The data set is $Z_{[n]} = \{ Z_{i}^{R_i}\}_{i \in [n]}$, and algorithm $W | Z_{[n]} \sim P_{W | Z_{[n]}}$. It then follows that $\sqrt{n}\cdot\gen(\xi,P_{W|Z_{[n]}}) = \Eb[X_{W}]$. The difference $X_{w} - X_{u}$ is
\begin{align}
X_{w} - X_{u} = \frac{1}{\sqrt{n}}\sum_{i = 1}^{n} R_i(\ell(w, Z_i^+) - \ell(u, Z_i^{+}) + \ell(u, Z_i^+) - \ell(w, Z_i^-) ).
\end{align}
By Hoeffding's lemma, $X_{w} - X_{u}$ given $Z^{\pm}_{[n]}$ is sub-Gaussian with  parameter 
\begin{align}
& \frac{1}{n}\sum_{i=1}^n(\ell(w, Z_i^+) - \ell(u, Z_i^{+}) + \ell(u, Z_i^+) - \ell(w, Z_i^-) )^2 \notag \\
& \leq \frac{2}{n}  \sum_{i=1}^n(\ell(w, Z_i^+) - \ell(u, Z_i^{+}))^2 + \frac{2}{n} \sum_{i=1}^n(\ell(w, Z_i^-) - \ell(u, Z_i^{-}))^2. \label{eqn:sqr-norm}
\end{align}
Define the RHS of the inequality above as $d^{2}_{Z_{[n]}^{\pm}}(w, u)$. % and define $d_k^2(u, v) := \Eb[d^{2}_{Z_{[n]}^{\pm}}(u, v) | W_{k}=u, W_{k-1}=v]$. 
Given $Z_{[n]}^{\pm}$, if a stochastic chain $\{W_{k}\}_{k \geq k_0}$ exists for the conditional process $X_{\Tc} | Z_{[n]}^{\pm}$, we know
\begin{align}
	\Eb[X_{W} | Z_{[n]}^{\pm}] &= \sum_{k = k_0+1}^\infty \Eb[X_{W_{k}} - X_{W_{k-1}} | Z_{[n]}^{\pm}] \\
	&\leq \sum_{k = k_0+1}^{\infty}  \inf_{\lambda > 0} \left\{ \frac{I_{Z_{[n]}^\pm}(X_{\Wc}; W_{k}, W_{k-1})}{\lambda} + \frac{\lambda}{2} \Eb[d_{Z^{\pm}_{[n]}}^2(W_{k}, W_{k-1}) | Z_{[n]}^{\pm}] \right\} \\
	& \leq \sum_{k = k_0+1}^{\infty}  \sqrt{2 I_{Z_{[n]}^{\pm}}(X_\Wc; W_k, W_{k-1}) } \sqrt{\Eb[d^2_{Z_{[n]}^\pm}(W_{k}, W_{k-1}) | Z_{[n]}^\pm]} \\
	& = \sum_{k = k_0+1}^{\infty} \sqrt{\Eb[d^2_{Z_{[n]}^\pm}(W_{k}, W_{k-1}) | Z_{[n]}^\pm]} \sqrt{2 I_{Z_{[n]}^{\pm}}(X_\Wc; W_k) },
\end{align}
where $I_{Z_{[n]}^\pm}(X_{\Wc}; W_{k}):=D(P_{X_{\Wc}, W_{k} | Z_{[n]}^{\pm}} || P_{X_{\Wc} | Z_{[n]}^{\pm}} \otimes P_{W_{k} | Z_{[n]}^{\pm}} )$, and the last equality is due to Markov chain $X_{\Wc} \leftrightarrow W_{k} \leftrightarrow W_{k-1}$.

To prove Theorem \ref{theorem:vc}, it suffices to construct a stochastic chain $\{W_k\}_{k=k_0}^{\infty}$ on the metric space $(\Wc, d_{Z_{[n]}^{\pm}}(\cdot, \cdot))$ for each given $Z_{[n]}^{\pm} = z_{[n]}^{\pm} \in \Zc^{2n}$, and show $\Eb[X_{W} | Z_{[n]}^{\pm}] = O(\sqrt{d_{VC}(\Wc)/n})$. The following analysis is performed given $Z_{[n]}^{\pm}$ and we simply write $d(\cdot, \cdot) = d_{Z_{[n]}^{\pm}}(\cdot, \cdot)$. For any metric space $(\Wc, \tau)$, denote $\Cc(\Wc, \tau, \epsilon)$ as a minimum set covering (minimum $\epsilon$-net) of the metric space $(\Wc, \tau)$ at scale $\epsilon > 0$, and define $N(\Wc, \tau, \epsilon) := |\Cc(\Wc, \tau, \epsilon)|$ as the corresponding covering number. 

According to the definition $d(\cdot, \cdot)$ in (\ref{eqn:sqr-norm}), we know that $\min\{ d(u, v) > 0: \forall u, v \in \Wc \} \geq \sqrt{\frac{2}{n}}$. Thus there exists $k_1 > 0$, such that $N(\Wc, d, 2^{-k_1}) = N(\Wc, d, \epsilon), \forall 0 < \epsilon < 2^{-k_1}$. 

We construct a stochastic chain as follows. Consider a sequence of countable sets  $\{\Pc_k\}_{k \geq k_0}^{k_1}$, where $\Pc_k = \Cc(\Wc, d, 2^{-k})$ and $N(\Wc, d, 2^{-k_0}) = 1$. Then we can construct a Markov process $W_{k_0} \leftrightarrow W_{k_0 + 1} \leftrightarrow \cdots \leftrightarrow W_{k_1}$, where $W_{k} = \argmin_{w \in \Pc_k} d(w, W_{k+1}) $ for each $k = k_0, \ldots, k_1 - 1$ and $W_{k_1} = \argmin_{w \in \Pc_{k_1}} d(w, W)$. We then know that $\Eb[X_{W} | Z_{[n]}^{\pm}] = \Eb[X_{W_{k_1}} | Z_{[n]}^{\pm}]$. It was known \cite[Theorem 8.3.18]{vershynin2018high} that there exists some universal constant $C > 0$ that $N(\Wc, d, \epsilon) \leq (C/\epsilon)^{C d_{VC}(\Wc)}$ for any $\epsilon$, which implies that $|\Pc_k| = N(\Wc, d, 2^{-k}) \leq 2^{C' k d_{VC}(\Wc)}$ for some universal constant $C'$. It follows that
\begin{align}
\Eb[X_{W} | Z_{[n]}^{\pm}] & \leq \sum_{k = k_0+1}^{k_1} \sqrt{\Eb[d^2_{Z_{[n]}^\pm}(W_{k}, W_{k-1}) | Z_{[n]}^\pm]} \sqrt{2 I_{Z_{[n]}^{\pm}}(X_\Wc; W_k) } \\
& \leq \sum_{k = k_0+1}^{k_1} 2^{-k+1} \sqrt{2 \ln(|\Pc_k|) } \\
& \leq \sum_{k = k_0+1}^{k_1} 2^{-k+1} \sqrt{2 C' k d_{VC}(\Wc)\ln(2) }  = O\left( \sqrt{d_{VC}(\Wc)} \right).
\end{align}
We thus have
\begin{align}
    \gen(\xi,P_{W|Z_{[n]}}) &=  \frac{1}{\sqrt{n}} \Eb[X_{W}] =  \frac{1}{\sqrt{n}}\Eb[ \Eb[X_{W} | Z_{[n]}^{\pm}]] = O\left(\sqrt{\frac{d_{VC}(\Wc)}{n}}\right),
\end{align}
which is the desired result.

\section{Details for the Gaussian setting via Theorem \ref{theorem:firstgen}} \label{subsection:firstgaussian}

Recall that 
\begin{align}
\Delta_k=\gen_{Z_{[n]}}(\xi,W_k)- \gen_{Z_{[n]}}(\xi,W_{k-1})= 2(W_{k-1} -W_k) \left(\bar{Z}_n - \mu \right).
\end{align}
Consequently we define the metric as 
\begin{align}
d^2(w, v) = \frac{4\sigma^2 (w-v)^2}{n},
\end{align}
and it can be verified that the process $\gen_{Z_{[n]}}(\xi,w)_{w\in\mathbb{R}}$ is indeed sub-Gaussian under this metric. Recall $W_{k} - \mu = \alpha_k(W - \mu + N_k)$, $W$ and $N_k$ are independent, and $\alpha_k =\frac{\sigma^2/n}{\sigma^2/n + \sigma_k^2} = \frac{1}{1 + 2^{-k}}$. It follows that $W_{k}$ has mean $\mu$ and variance $ \alpha_k\frac{\sigma^2}{n}$. Since $\frac{\alpha_{k-1}}{\alpha_{k}} -1 = \frac{-1}{2^k + 2}$, by the relations between $W_k$ and $W_{k-1}$ in (\ref{eqn:W_k-1_k}) we can calculate
\begin{align}
\Eb[d(W_k, W_{k-1})^2] & = \frac{4 \sigma^2}{n} \Eb\left[ (W_k - W_{k-1})^2 \right] \\
&= \frac{4 \sigma^2}{n} \Eb\left[\left( \left(\frac{\alpha_{k-1}}{\alpha_{k}} - 1 \right) (W_k - \mu) + \alpha_{k-1} N'_{k} \right)^2 \right]\\
&= \frac{4 \sigma^2}{n} \Eb\left[\left( \frac{-1}{2^k + 2} (W_k - \mu) + \frac{1}{1 + 2^{-k+1}} N'_{k} \right)^2 \right]\\
&= \frac{4 \sigma^2}{n} \left( \frac{1}{(2^k+2)^2} \Eb[(W_k - \mu)^2] +  \frac{1}{(1 + 2^{-k+1})^2} \Eb[(N'_k)^2] \right) \\
&= \frac{4 \sigma^2}{n} \left( \frac{\alpha_k}{(2^k+2)^2} +  \frac{1}{(2^k + 2)(1 + 2^{-k+1})} \right) \frac{\sigma^2}{n} \\
&< \frac{4 \sigma^2}{n} \left( \frac{1/2}{2^k+2} +  \frac{1}{2^k + 2} \right) \frac{\sigma^2}{n} \\
&= \frac{4 \sigma^4}{n^2} \frac{3/2}{2^k + 2} = \frac{\sigma^4}{n^2} \frac{3}{2^{k-1} + 1},
\end{align}
where the inequality is by $\frac{\alpha_k}{2^k + 2} = \frac{1}{2 + 2^k + 1 + 2^{-k+1}} < \frac{1}{2}$ and $\frac{1}{1+ 2^{-k+1}} < 1$. Furthermore, the mutual information can be upper bounded as
\begin{align}
I(Z_{[n]}; W_k) \leq I(W; W_{k}) = \frac{1}{2} \ln(1 + 2^k),
\end{align}
where the inequality is due to the data processing inequality over the Markov chain $Z_{[n]} \leftrightarrow W \leftrightarrow W_k$ and the equality is by the Gaussian channel nature of the stochastic chain design. By the mutual information based generalization error bound in Theorem \ref{theorem:firstgen}, we have
\begin{align}
 \Eb[X_W]  & \leq \sum_{k = -\infty}^{\infty} \sqrt{\Eb[d^2(W_k,W_{k-1})]}\sqrt{2I(Z_{[n]}; W_k)} \\
 & \leq \frac{\sigma^2}{n} \sum_{k = -\infty}^{\infty} \sqrt{\frac{3 \ln(1 + 2^k)}{2^{k-1} + 1}} < 13 \frac{\sigma^2}{n},
\end{align}
where $ \sum_{k = -\infty}^{\infty} \sqrt{\frac{3 \ln(1 + 2^k)}{2^{k-1} + 1}} < 13$ is calculated numerically. \qed

\section{Details of the Gaussian setting via Theorem \ref{theorem:secondgen}}

In this section, we derived a tightened generalization bound for the Gaussian setting using Theorem \ref{theorem:secondgen}. We fixed an $i \in [n]$ in the discussion below, since in this setting there is no material difference in the indices. 
\begin{align}
\Delta \gen^i_k=\gen^i_{Z_i}(\xi,W_{i,k})-\gen^i_{Z_i}(\xi,W_{i, k-1}) = 2(W_{i, k-1} - W_{i,k}) (Z_i - \mu).
\end{align}
Consequently we define the metric as 
\begin{align}
d^2_i(w, v) = 4\sigma^2 (w - v)^2.
\end{align}
and it can be verified that the process $\gen^i_{Z_i}(\xi, w)_{w \in \Rb}$ is indeed sub-Gaussian under this metric. 
Select a sequence of mutually independent Gaussian noise $\{N'_{k}\}_{i \in \Nb}$, which is independent of $W$ and $N'_k \sim \Nc(0, \sigma'^2_k)$, where $\sigma'^2_k = \frac{\sigma^2}{2^k n}$. Define the cumulative noise 
\begin{align}
N_k := \sum_{j = k+1}^\infty N'_j \sim \Nc(0, \sigma^2_k),
\end{align}
where $\sigma^2_k = \frac{\sigma^2}{2^k n}$. The stochastic chain is designed as 
\begin{align}
W_{i,k} - \mu = \alpha_k (W - \mu + N_k),
\end{align}
where $\alpha_k = \frac{\sigma^2}{\sigma^2 + n \sigma_k^2} = \frac{1}{1 + 2^{-k}}$. We then have
\begin{align}
W_{i, k-1} - \mu = \frac{\alpha_{k-1}}{\alpha_k}(W_{i, k} - \mu) + \alpha_{k-1} N_k'.
\end{align}
Under this stochastic chain, similar to the proof in the previous section, we can derive the expression for $\sqrt{\Eb[d(W_{i, k}, W_{i, k-1})^2]}$ as
\begin{align}
\Eb[d(W_{i, k}, W_{i, k-1})^2] & = 4 \sigma^2 \Eb\left[ (W_{i, k}, W_{i, k-1})^2 \right] 
< \frac{\sigma^4}{n} \frac{3}{2^{k-1} + 1}.
\end{align}
Since
\begin{align}
W_{i, k} - \mu = \alpha_k (W - \mu + N_k) = \alpha_k \left( \frac{1}{n}(Z_i - \mu) + \frac{1}{n}\sum_{j \not= i} (Z_j - \mu) + N_k \right),
\end{align}
the mutual information term $I(Z_{i}; W_{i,k})$ can be calculated as
\begin{align}
I(Z_{i}; W_{i,k}) & = \frac{1}{2}\ln\left( 1 + \frac{\sigma^2/n^2}{ (n-1)\sigma^2/n^2 + \sigma^2_k} \right) \\
& = \frac{1}{2}\ln\left( 1 + \frac{1}{n-1 + n 2^{-k}} \right) \\
&= - \frac{1}{2}\ln\left( 1 - \frac{1}{n (1 + 2^{-k})} \right).
\end{align}
By the convexity of function $-\ln(1 - x)$, we know for any $x \in (0, 1)$,
\begin{align}
- \ln\left( 1 - \frac{x}{n} \right) \leq - \frac{n-1}{n} \ln(1 - 0) - \frac{1}{n} \ln(1 - x) = -\frac{1}{n} \ln(1 - x).
\end{align}
Let $x = 1/(1 + 2^{-k})$, we have $I(Z_i; W_{i,k}) \leq \frac{1}{2n} \ln(1 + 2^{k})$. Since $n \geq 2$, we know for any $x \in (0, 1)$
\begin{align}
-\ln\left(1 - \frac{x}{n}\right) < -\ln\left(1 - \frac{1}{n}\right)  \leq - \frac{n-2}{n}\ln(1 - 0) - \frac{2}{n} \ln\left(1 - \frac{1}{2}\right) \leq \frac{2 \ln(2)}{n},
\end{align}
which implies $I(Z_i ; W_{k}) \leq \frac{\ln(2)}{n}$.

By the mutual information based generalization error bound in Theorem \ref{theorem:secondgen}, we have
 the mutual information based bound is
\begin{align}
    \gen(\xi,P_{W|Z_{[n]}}) & \leq  \frac{1}{n}\sum_{i=1}^n \sum_{k= -\infty}^\infty \sqrt{\Eb[d^2(W_{i,k},W_{i,k-1})]}\sqrt{2I(Z_{i}; W_{i,k})} \\
    & \leq \frac{\sigma^2}{n} \sum_{k = -\infty}^{\infty} \sqrt{\frac{3 \min(\ln(1 + 2^k), 2\ln(2))}{2^{k-1} + 1}} < 11\frac{\sigma^2}{n},
\end{align}
where $  \sum_{k = -\infty}^{\infty} \sqrt{\frac{3 \min(\ln(1 + 2^k), 2\ln(2))}{2^{k-1} + 1}} < 11$ is calculated numerically. \qed

\end{document}